\documentclass[
reprint,
aps,
pra,
superscriptaddress,
longbibliography,
nofootinbib,
floatfix
]{revtex4-2}

\usepackage[T1]{fontenc}
\usepackage[utf8]{inputenc}

\usepackage{
	amsmath,
	amssymb,
	amsfonts,
	amsthm,
	mathtools,
	bm
}

\usepackage{braket}

\usepackage{graphicx}
\usepackage{tikz}

\usetikzlibrary{
	arrows.meta,
	positioning,
	fit,
	trees,
	shapes,
	calc
}

\usepackage{booktabs}
\usepackage{tabularx}
\usepackage{dcolumn}

\usepackage{hyperref}

\theoremstyle{plain}
\newtheorem{theorem}{Theorem}

\theoremstyle{definition}
\newtheorem{definition}{Definition}

\newcommand{\OO}{\mathcal{O}}

\newcommand{\EE}{\mathcal{E}}

\begin{document}
	
	
	\title{
		Emergent Operational Entanglement Graphs and Sub-Quadratic Authentication Scaling in Realistic E91 Quantum Networks
	}
	
	\author{Jose Luis Rosales}
	\email{JoseLuis.Rosales@upm.es}
	
	\affiliation{
		Universidad Polit\'ecnica de Madrid (UPM)\\
		Center for Computational Simulation\\
		DLSIIS, ETSI Inform\'aticos\\
		Boadilla del Monte, Spain
	}
	
	\date{\today}
	
	
	\begin{abstract}
		
		Large-scale entanglement-based quantum key distribution (QKD) networks are commonly assumed to require authentication resources scaling quadratically with the number of users. We show that realistic quantum communication networks operating under loss, decoherence, and LOCC constraints exhibit fundamentally different scaling laws. Using Pauli transfer matrix (PTM) transport, we demonstrate that Bell correlations decay exponentially along entanglement-swapping paths, generating finite operational correlation lengths and sparse operational entanglement graphs. In sparse metropolitan quantum networks, the number of CHSH-usable Bell pairs consequently scales linearly with network size, while authentication complexity scales as
		\[
		\Theta(N\log N),
		\]
		under sparse-mixing assumptions. We further present an ancilla-assisted distributed Bell-state verification framework for realistic E91 quantum metropolitan infrastructures. Our results suggest that scalable authentication in quantum communication networks emerges directly from the physics of entanglement transport.
		
	\end{abstract}
	\maketitle
	
	\section{Introduction}
	
	The development of scalable quantum communication infrastructures is one of the central objectives of modern quantum technologies. The long-term vision of a quantum internet involves distributed entanglement generation among geographically separated nodes connected through optical fibre and satellite infrastructures \cite{Kimble2008,Wehner2018,Azuma2023}. 
	
	Recent experimental deployments have demonstrated significant progress toward this objective. Metropolitan-scale infrastructures such as MadQCI and EuroQCI illustrate the feasibility of heterogeneous quantum-secured communication systems operating over realistic production environments\cite{martin2024}. In parallel, recent entanglement-swapping demonstrations over urban fibre networks further support the practical viability of distributed Bell-state transport across realistic telecommunication infrastructures.
	
	Entanglement-based QKD protocols, particularly the E91 protocol, provide device-independent security guarantees rooted in Bell inequality violations \cite{Acin2007DIQKD}. However, the scalability of authentication in large multi-user quantum networks remains poorly understood. Conventional analyses frequently assume that authentication resources scale quadratically with the number of users because every pair of communicating nodes potentially requires independent authentication resources \cite{diamanti2016}. 
	
	This perspective implicitly treats network topology as externally imposed and assumes that every pair of nodes remains equally accessible from an operational viewpoint. Realistic quantum networks, however, are fundamentally constrained by physical transport phenomena. Photon loss, imperfect Bell-state measurements, finite quantum-memory coherence times, probabilistic entanglement swapping, purification overheads, and measurement-conditioned LOCC protocols strongly constrain the spatial propagation of usable entanglement \cite{Watrous2018}.
	
	Consequently, the operational connectivity of a quantum network differs substantially from its underlying physical topology. While a metropolitan fibre infrastructure may exhibit global small-world connectivity, the subset of Bell correlations that remain operationally useful for E91 authentication can become dramatically sparse due to finite correlation transport lengths \cite{Watts1998}.
	
	The central thesis of the present work is that authentication complexity in realistic E91 quantum networks is not governed by combinatorial pair counting but instead emerges dynamically from the physics of entanglement transport under realistic operational constraints.
	
	Unlike conventional graph-theoretic approaches, the present work treats operational Bell connectivity as an emergent physical property determined by correlation transport under realistic network conditions. In practical quantum communication infrastructures, Bell correlations propagate through lossy entanglement-swapping chains subject to decoherence, finite quantum-memory coherence times, imperfect Bell-state measurements, and measurement-conditioned LOCC operations. These mechanisms dynamically constrain the spatial extent of operationally usable Bell correlations, generating finite operational correlation lengths analogous to screening phenomena in many-body systems.
	
	Within this framework, the effective operational entanglement graph differs fundamentally from the underlying physical fibre topology. Although metropolitan quantum networks may exhibit global small-world connectivity, the subset of Bell pairs that remains operationally usable for E91 authentication becomes intrinsically sparse. As a consequence, authentication complexity is governed primarily by the transport physics of entanglement rather than by naive combinatorial pair counting. This emergent perspective establishes a direct connection between quantum-network topology, Bell-correlation transport, entanglement percolation, and scalable authentication architectures in realistic quantum communication infrastructures.
	
	Our approach differs conceptually from most existing work on quantum networking. Previous studies typically optimize entanglement distribution over a predefined network graph. Here we instead derive the effective operational graph from the physical propagation of Bell correlations themselves \cite{Acin2007,Cuquet2009,Pant2019}.
	
	We show that finite operational correlation lengths induce sparse emergent operational entanglement graphs whose number of CHSH-usable Bell pairs scales linearly with network size despite logarithmic global network connectivity. This naturally produces sub-quadratic authentication scaling laws.
	
	More broadly, the results support an interpretation in which realistic quantum communication infrastructures behave as distributed many-body systems whose effective topology emerges dynamically from correlation transport under loss and decoherence \cite{Hastings2006,Verstraete2008}.

The present work develops a complementary physical framework in which authentication complexity emerges directly from the transport properties of entanglement itself. Rather than treating Bell connectivity as an externally imposed graph-theoretic resource, we derive the effective operational entanglement structure from the propagation of Bell correlations under realistic lossy quantum-network dynamics. In particular, we show that Pauli-transfer-matrix (PTM) transport naturally generates finite operational correlation lengths due to loss, decoherence, imperfect Bell measurements, and LOCC-conditioned entanglement swapping. This mechanism dynamically restricts the set of operationally usable Bell pairs to sparse local neighbourhoods, even when the underlying physical infrastructure possesses global small-world connectivity \cite{Watts1998}.

From this perspective, realistic E91 quantum communication networks exhibit an emergent operational geometry fundamentally different from their physical fibre topology. The resulting CHSH-usable Bell graph becomes intrinsically sparse, with each node maintaining only a finite number of operational Bell partners independent of total network size. This physical sparsification mechanism immediately modifies authentication scaling laws. Under physically motivated sparse-mixing assumptions, we show that the total authentication overhead scales as
\[
\Theta(N\log N),
\]
rather than quadratically. More generally, the present work establishes a direct connection between entanglement transport physics, operational network topology, and scalable authentication in realistic quantum communication infrastructures.

Beyond the scaling analysis itself, the work also introduces a physically motivated operational interpretation of quantum metropolitan networks as distributed entanglement-transport systems. In this interpretation, operational Bell connectivity emerges dynamically from correlation propagation under realistic constraints in close analogy with finite-correlation and screening phenomena in many-body physics. The framework therefore connects quantum-network theory, open quantum systems, entanglement transport, and operational QKD architectures within a unified description.

The remainder of this article is organised as follows. In Sec.~II we introduce the physically constrained quantum-network model and define operational Bell connectivity. Section~III develops the PTM description of multi-hop entanglement transport and derives the emergence of finite operational correlation lengths. Section~IV analyses the resulting sparse operational entanglement graph and proves the linear scaling of CHSH-usable Bell pairs. In Sec.~V we derive the corresponding authentication complexity under sparse-mixing assumptions. Section~VI presents the ancilla-assisted distributed Bell-state verification framework and discusses its operational role in realistic E91 infrastructures. Finally, Sec.~VII discusses the broader physical interpretation of emergent operational network geometry and Sec.~VIII summarizes the main conclusions and implications for scalable quantum communication architectures.
	
	\section{Physically Constrained Quantum Communication Networks}
	
	We model a metropolitan quantum communication infrastructure as a sparse graph
	\[
	G=(V,E),
	\]
	where the vertices $V$ represent quantum nodes and the edges $E$ correspond to elementary optical-fibre links capable of distributing entangled photonic states. The total number of network nodes is denoted by
	\[
	N=|V|.
	\]
	
	Our analysis focuses on realistic metropolitan-scale quantum networks operating in the sparse-connectivity regime, where the average number of physical neighbours per node remains bounded independently of network size,
	\[
	\langle k\rangle=\OO(1).
	\]
	At the same time, the network retains small-world characteristics, so that the average shortest-path distance between arbitrary nodes grows only logarithmically with the total number of nodes,
	\[
	\langle d\rangle=\OO(\log N).
	\]
	
	This class of topologies provides a realistic abstraction of optical metropolitan infrastructures and quantum communication backbones, where most nodes possess only local fibre connectivity while a smaller number of longer-range links or routing shortcuts significantly reduce global communication distances. Sparse small-world architectures are particularly relevant in practical quantum networks because they combine scalability, moderate physical deployment cost, and efficient long-range entanglement distribution.
	
	\subsection{Operational Bell Connectivity}
	
	The key distinction introduced in this work is the separation between:
	
	\begin{enumerate}
		\item the underlying physical network graph,
		\item the operational Bell-connectivity graph.
	\end{enumerate}
	
	\begin{definition}[Operational Bell link]
		A Bell link between nodes $A$ and $B$ is operational if the distributed state supports Bell correlations strong enough to enable E91-type authentication and QKD under realistic physical constraints.
	\end{definition}
	
	Operational connectivity therefore depends not only on graph topology but also on:
	
	\begin{itemize}
		\item channel transmissivity,
		\item memory decoherence,
		\item Bell-state measurement fidelity,
		\item LOCC success probability,
		\item entanglement purification overhead,
		\item swapping depth.
	\end{itemize}
	
	Consequently, operational Bell connectivity becomes a dynamical physical observable rather than a purely combinatorial property of the graph.
	
	\section{PTM Correlation Transport}
	
	To describe the transport of Bell correlations through realistic swapping chains we employ the Pauli transfer matrix (PTM) representation \cite{Watrous2018}.
	
	Let:
	\[
	\EE:\mathcal B(\mathcal H)\rightarrow\mathcal B(\mathcal H)
	\]
	denote a completely positive map acting on density operators.
	
	In the PTM representation:
	\[
	\mathrm{vec}(\rho')
	=
	T^{(\EE)}
	\mathrm{vec}(\rho),
	\]
	where:
	\[
	T^{(\EE)}
	\]
	is the Pauli transfer matrix associated with the channel.
	
	\subsection{Multi-Hop Entanglement Propagation}
	
	Consider a multi-hop entanglement-distribution process connecting two distant nodes through a sequence of entanglement-swapping operations along a network path of length $L$. Each elementary step of the propagation process involves a combination of Bell-state measurements, local unitary corrections, noisy quantum channels, and measurement-conditioned local operations and classical communication (LOCC). As entanglement propagates through the network, these operations collectively modify the transported quantum correlations and progressively reduce their operational visibility.
	
	To describe this transport process, we employ the Pauli transfer matrix (PTM) representation of quantum channels. Let
	\[
	T^{(\Phi_i)}
	\]
	denote the PTM associated with the $i$-th elementary propagation step. The effective transport channel associated with a complete entanglement-swapping path
	\[
	P=(\Phi_1,\Phi_2,\ldots,\Phi_L)
	\]
	is then obtained through the sequential composition
	\begin{equation}
		T^{(\mathcal N_P)}
		=
		T^{(\Phi_L)}
		\cdots
		T^{(\Phi_1)}.
		\label{eq:ptmchain}
	\end{equation}
	
	Under realistic physical conditions, including optical loss, imperfect Bell measurements, memory decoherence, and probabilistic heralding, the spectral norm of each elementary transport channel remains strictly contractive,
	\[
	\|T^{(\Phi_i)}\|_2\le\lambda_{\max}<1.
	\]
	Consequently, the norm of the effective multi-hop transport channel decays exponentially with propagation depth,
	\begin{equation}
		\|T^{(\mathcal N_P)}\|_2
		\le
		\lambda_{\max}^{L}.
		\label{eq:ptmdecay}
	\end{equation}
	
	Equation~(\ref{eq:ptmdecay}) expresses the central physical mechanism underlying the present work: Bell correlations are progressively attenuated as entanglement propagates through successive swapping operations. This behaviour is closely analogous to finite-correlation and screening phenomena in many-body systems, where long-range correlations decay exponentially beyond a characteristic transport scale \cite{Hastings2006}.
	
	To quantify the operational strength of the transported Bell correlations, we introduce the Bell-correlation visibility
	\begin{equation}
		C_{\mathrm{Bell}}(L)
		=
		C_0\,\lambda_{\max}^{L},
		\label{eq:bellcorr}
	\end{equation}
	where $C_0$ denotes the initial Bell-correlation visibility of the elementary links. 
	
	For realistic metropolitan fibre links, the distributed states are well approximated by Werner-like Bell-diagonal states of the form
	\begin{equation}
		\rho_w
		=
		w\ket{\Phi^+}\bra{\Phi^+}
		+
		(1-w)\frac{\mathbb I}{4},
	\end{equation}
	where $w$ characterizes the Bell-state visibility. For this class of states, violation of the CHSH Bell inequality occurs when \cite{Horodecki1995}
	\[
	w>\frac1{\sqrt2}.
	\]
	
	Operational Bell correlations therefore remain usable only up to a finite maximum swapping depth. Combining the exponential PTM decay with the CHSH threshold condition yields the maximum operational propagation depth
	\begin{equation}
		L_{\max}
		=
		\left\lfloor
		\frac{
			\log(1/\sqrt2 C_0^{-1})
		}{
			\log(\lambda_{\max})
		}
		\right\rfloor.
		\label{eq:lmax}
	\end{equation}
	
	Because the effective transport channels remain contractive,
	\[
	\lambda_{\max}<1,
	\]
	the operational correlation length remains finite even in globally connected networks,
	\[
	L_{\max}=\OO(1).
	\]
	
	This finite operational correlation length plays a central role throughout the remainder of the paper, since it dynamically constrains the subset of Bell pairs that remain physically usable for E91 authentication in realistic quantum communication infrastructures.
	
	This finite-correlation mechanism plays a role analogous to screening phenomena in condensed-matter systems: correlations may propagate globally through the network topology, yet operational Bell nonlocality becomes spatially localized \cite{Hastings2006}.
	
	\section{Emergent Operational Entanglement Graphs}
	
	For a node:
	\[
	A\in V,
	\]
	define the operational Bell neighbourhood:
	\begin{equation}
		\mathcal S_A
		=
		\{B\in V:d(A,B)\le L_{\max}\}.
	\end{equation}
	
	In sparse small-world graphs, the number of nodes accessible within distance $L$ scales asymptotically as \cite{Watts1998,Cuquet2009}:
	\[
	|\mathcal B_L(A)|
	=
	\OO((k-1)^L).
	\]
	
	Since:
	\[
	k=\OO(1),
	\qquad
	L_{\max}=\OO(1),
	\]
	we obtain:
	\[
	|\mathcal S_A|=\Theta(1).
	\]
	
	Therefore, every node possesses only a finite number of operational Bell partners independent of total network size.
	
	This leads directly to the following result.
	
	\begin{theorem}[Linear scaling of operational Bell pairs]
		The total number of CHSH-usable Bell pairs in the network scales linearly with network size:
		\[
		|\mathcal S(N)|=\Theta(N).
		\]
	\end{theorem}
	
	\begin{proof}
		
		Each node possesses at most:
		\[
		C(k,L_{\max})
		\]
		operational Bell partners.
		
		Summing over all nodes:
		\[
		|\mathcal S(N)|
		\le
		N\,C(k,L_{\max})
		=
		\Theta(N).
		\]
		
	\end{proof}
	
	The operational entanglement graph therefore remains sparse even when the underlying physical network exhibits logarithmic global connectivity.
	
	This emergent sparsification constitutes the key mechanism responsible for sub-quadratic authentication complexity \cite{Acin2007,Pant2019}.
	
	\section{Authentication Complexity}
	
	Authentication requires discovering operational Bell partners within physically accessible neighbourhoods.
	
	Unlike classical fully connected communication graphs, realistic E91 quantum networks cannot operationally exploit arbitrary node pairs because Bell correlations become exponentially suppressed beyond finite correlation depth \cite{Pirandola2017,Azuma2023}.
	
	\subsection{Sparse-Mixing Authentication Model}
	
Our analysis assumes a physically realistic sparse-connectivity regime in which operational Bell links remain locally distributed throughout the network. In particular, we consider metropolitan quantum infrastructures where the number of operationally usable Bell partners associated with each node remains finite, while neighbourhood exploration retains logarithmic small-world scaling. We additionally assume that local Bell-verification attempts are approximately statistically independent over the relevant operational neighbourhood scale.

Under these conditions, each node explores a logarithmic neighbourhood of size
\[
n=\Theta(\log N),
\]
corresponding to the characteristic small-world exploration depth of the network. Because the number of operational Bell partners remains bounded independently of total network size,
\[
|\mathcal S_A|=\Theta(1),
\]
the probability of successfully identifying a CHSH-usable Bell partner remains asymptotically finite,
\[
P_{\mathrm{succ}}=\Theta(1).
\]

This regime captures the physically relevant situation in which operational Bell connectivity remains sparse despite the existence of global network reachability through small-world transport paths.
	
	\subsection{Sub-Quadratic Authentication Scaling}
	
	\begin{theorem}[Authentication complexity]
		Under sparse-mixing assumptions, authentication complexity in realistic E91 small-world quantum networks scales as
		\[
		\Theta(N\log N).
		\]
	\end{theorem}
	
	\begin{proof}
		
		Each node performs:
		\[
		\Theta(\log N)
		\]
		local Bell-verification attempts. Summing over all nodes:
		\[
		N\times\Theta(\log N)
		=
		\Theta(N\log N).
		\]
		\end{proof}
	The logarithmic overhead originates from probabilistic Bell-partner discovery in sparse operational neighbourhoods. Importantly, this scaling law emerges directly from finite-correlation transport physics rather than from externally imposed graph-theoretic assumptions.
	
\section{Distributed Bell-State Verification}

We now discuss the physical verification of distributed Bell correlations.

Consider two remote nodes sharing an unknown bipartite state
\[
\rho_{AB},
\]
together with a trusted Bell reference pair
\[
\ket{\Phi^+}_{A'B'}.
\]

Ancilla-assisted distributed SWAP measurements allow Bell-state certification without directly measuring the primary distributed qubits. The protocol combines local Hadamard operations, controlled-SWAP interactions, and computational-basis measurements performed on ancillary qubits located at the remote network nodes.

Defining the Bell fidelity
\begin{equation}
	F
	=
	\bra{\Phi^+}
	\rho_{AB}
	\ket{\Phi^+},
\end{equation}
the probability of obtaining the ancilla outcome $(0,0)$ is
\begin{equation}
	P_{00}
	=
	\frac{1+F}{2}.
	\label{eq:p00}
\end{equation}

Since separable two-qubit states satisfy
\[
F\le\frac12,
\]
the condition
\[
P_{00}>\frac34
\]
constitutes a sufficient witness of entanglement.

The full derivation of the distributed SWAP-based verification protocol, including the evaluation of the SWAP expectation values and the separability bounds, is provided in the Supplemental Material.
	
	\subsection{Remarks on Clifford Structure}
	
	The verification protocol employs controlled-SWAP operations.
	
	Because the controlled-SWAP (Fredkin) gate is non-Clifford, a full Gottesman--Knill stabilizer simulation is not generally applicable. Nevertheless, Bell-diagonal inputs admit an effective reduced correlator description compatible with stabilizer propagation techniques \cite{Watrous2018}.
	
	The protocol should therefore be interpreted as a physically motivated Bell-certification primitive rather than as a purely stabilizer-based Clifford circuit.
	
	\section{Physical Interpretation}
	
	The results obtained here support a broader physical interpretation of realistic quantum communication infrastructures.
	
	In conventional communication theory, topology is externally designed while communication resources propagate over fixed graphs. In realistic quantum networks, however, operational connectivity itself becomes dynamically generated by the transport properties of entanglement.
	
Loss, decoherence, imperfect Bell-state measurements \cite{Azuma2023}, and finite quantum-memory coherence times induce finite operational correlation lengths analogous to screening lengths in many-body physics \cite{Hastings2006,Verstraete2008}. As a consequence, the effective operational behaviour of the network differs fundamentally from its underlying physical topology. Although the physical fibre infrastructure may remain globally connected and retain small-world communication properties, the subset of Bell correlations that remains operationally usable for E91 authentication becomes intrinsically sparse. Authentication resources therefore scale sub-quadratically with network size, while the effective operational topology itself emerges dynamically from the transport and attenuation of Bell correlations across the network.

From this perspective, scalable quantum communication infrastructures may be interpreted as distributed nonequilibrium many-body systems whose effective geometry is not statically imposed by the physical graph, but instead emerges from the collective propagation of entanglement under realistic transport constraints.
	
	\section{Conclusions}
	
	In this work we have shown that authentication complexity in realistic E91 quantum communication networks is fundamentally governed by the physics of entanglement transport rather than by naive combinatorial pair counting. Using a PTM-based description of multi-hop Bell-correlation propagation, we demonstrated that loss, decoherence, imperfect swapping operations, and LOCC-conditioned transport induce finite operational correlation lengths that dynamically constrain the set of physically usable Bell pairs. As a consequence, realistic quantum metropolitan infrastructures generate sparse emergent operational entanglement graphs even when the underlying fibre topology exhibits global small-world connectivity.
	
	The central contribution of the paper is the demonstration that this physical sparsification mechanism naturally leads to sub-quadratic authentication scaling laws. In sparse metropolitan quantum networks with bounded operational degree, the number of CHSH-usable Bell pairs scales linearly with network size, while authentication overhead acquires only a logarithmic correction associated with probabilistic Bell-partner discovery. Under sparse-mixing assumptions, the resulting authentication complexity scales as
	\[
	\Theta(N\log N).
	\]
	
	More broadly, the present framework suggests a new physical interpretation of scalable quantum communication infrastructures as distributed entanglement-transport systems whose effective operational geometry emerges dynamically from correlation propagation under realistic constraints. This perspective establishes a direct conceptual bridge between quantum-network theory, open-system correlation transport, entanglement percolation, and scalable E91 authentication architectures, providing a physically grounded route toward realistic large-scale quantum-secured communication networks \cite{Acin2007,Cuquet2009,Pant2019}.
	\section*{Acknowledgments}
	
	This work has been supported by the ``Hub Nacional de Excelencia en
	Comunicaciones Cuánticas'' project, funded by the Ministerio para la
	Transformación Digital y de la Función Pública within the framework of the
	Plan de Recuperación, Transformación y Resiliencia and the Mecanismo de
	Recuperación y Resiliencia. Funded by the European Union --
	NextGenerationEU.
	
	\appendix
	
	\section{Distributed SWAP-Based Bell-State Verification}
	
	In this Appendix we derive the distributed ancilla-assisted SWAP verification protocol used in the main text and establish the corresponding separability bound.
	
	\subsection{Protocol Definition}
	
	Consider two remote nodes, Alice and Bob, sharing an unknown two-qubit state
	\[
	\rho_{AB},
	\]
	together with a trusted Bell reference pair
	\[
	\ket{\Phi^+}_{A'B'}
	=
	\frac{1}{\sqrt2}
	\left(
	|00\rangle+|11\rangle
	\right).
	\]
	
	Ancillary qubits
	\[
	a_A,\qquad a_B
	\]
	are initialized in the computational state
	\[
	|0\rangle_{a_A}\otimes |0\rangle_{a_B}.
	\]
	
	The global initial state is therefore
	\begin{equation}
		\rho_0
		=
		|00\rangle\langle00|_{a_A a_B}
		\otimes
		\rho_{AB}
		\otimes
		|\Phi^+\rangle\langle\Phi^+|_{A'B'}.
	\end{equation}
	
	Hadamard gates are first applied to both ancillas, preparing
	\[
	|+\rangle
	=
	\frac{|0\rangle+|1\rangle}{\sqrt2}.
	\]
	
	The ancilla subsystem becomes
	\[
	|+\rangle_{a_A}\otimes |+\rangle_{a_B}.
	\]
	
	Controlled-SWAP operations are then applied locally:
	\begin{align}
		U_A
		&=
		|0\rangle\langle0|_{a_A}\otimes I
		+
		|1\rangle\langle1|_{a_A}\otimes S_A,
		\\
		U_B
		&=
		|0\rangle\langle0|_{a_B}\otimes I
		+
		|1\rangle\langle1|_{a_B}\otimes S_B,
	\end{align}
	where:
	\[
	S_A\equiv S_{AA'},
	\qquad
	S_B\equiv S_{BB'}
	\]
	denote the SWAP operators acting on the corresponding subsystems.
	
	Finally, Hadamard gates are again applied to both ancillas and the ancillas are measured in the computational basis.
	
	\subsection{Joint Ancilla Probabilities}
	
	Define
	\[
	\sigma
	=
	\rho_{AB}
	\otimes
	|\Phi^+\rangle\langle\Phi^+|_{A'B'}.
	\]
	
	After the controlled-SWAP operations and the second Hadamard layer, the joint ancilla measurement probabilities satisfy
	\begin{widetext}
	\begin{equation}
		P_{m_A,m_B}
		=
		\frac14
		\Big[
		1
		+
		(-1)^{m_A}\langle S_A\rangle
		+
		(-1)^{m_B}\langle S_B\rangle
		+
		(-1)^{m_A+m_B}
		\langle S_A S_B\rangle
		\Big],
		\label{eq:jointprob_appendix}
	\end{equation}
	\end{widetext}
	where expectation values are evaluated with respect to $\sigma$ \footnote{
		The joint ancilla probability distribution is obtained by explicitly propagating the ancilla density operator through the two Hadamard layers and the controlled-SWAP operations, followed by a partial trace over the data registers. Writing
		\[
		\sigma=\rho_{AB}\otimes
		|\Phi^+\rangle\langle\Phi^+|_{A'B'},
		\]
		the controlled-SWAP operations imprint the expectation values of the SWAP operators onto the ancilla coherences. After the second Hadamard transformation, the resulting interference terms generate the parity structure appearing in Eq.~(\ref{eq:jointprob_appendix}). The full derivation follows from explicit expansion of the computational-basis amplitudes and repeated application of the identities
		\[
		S_A^2=S_B^2=I
		\]
		together with cyclicity of the trace.
	}.
	
	\subsection{Evaluation of the SWAP Expectation Values}
	
	Define the Bell fidelity
	\begin{equation}
		F
		=
		\bra{\Phi^+}
		\rho_{AB}
		\ket{\Phi^+}.
	\end{equation}
	
	Using the SWAP identity
	\[
	S_{AA'}
	(|i\rangle_A|k\rangle_{A'})
	=
	|k\rangle_A|i\rangle_{A'},
	\]
	one obtains
	\begin{equation}
		\langle S_A\rangle
		=
		\mathrm{Tr}[S_A\sigma]
		=
		F.
	\end{equation}
	
	By symmetry,
	\begin{equation}
		\langle S_B\rangle
		=
		F.
	\end{equation}
	
	Similarly, since the Bell reference state is invariant under subsystem exchange,
	\begin{equation}
		\langle S_A S_B\rangle
		=
		1.
	\end{equation}
	
	Substituting these results into Eq.~(\ref{eq:jointprob_appendix}) yields
	\begin{equation}
		P_{m_A,m_B}
		=
		\frac14
		\Big[
		1
		+
		(-1)^{m_A}F
		+
		(-1)^{m_B}F
		+
		(-1)^{m_A+m_B}
		\Big].
	\end{equation}
	
	Selecting the outcome
	\[
	m_A=m_B=0
	\]
	gives
	\begin{equation}
		P_{00}
		=
		\frac14(1+F+F+1)
		=
		\frac{1+F}{2}.
		\label{eq:p00_appendix}
	\end{equation}
	
	\subsection{Separability Bound}
	
	We now prove that any separable two-qubit state satisfies
	\[
	F\le\frac12.
	\]
	
	Consider first a pure product state
	\[
	|\psi\rangle_A\otimes |\varphi\rangle_B,
	\]
	with
	\[
	|\psi\rangle
	=
	\alpha|0\rangle+\beta|1\rangle,
	\qquad
	|\varphi\rangle
	=
	\gamma|0\rangle+\delta|1\rangle,
	\]
	and normalization conditions
	\[
	|\alpha|^2+|\beta|^2=1,
	\qquad
	|\gamma|^2+|\delta|^2=1.
	\]
	
	The overlap with the Bell state is
	\[
	\langle\Phi^+|\psi\varphi\rangle
	=
	\frac1{\sqrt2}
	(\alpha\gamma+\beta\delta).
	\]
	
	The corresponding Bell fidelity becomes
	\begin{equation}
		F_{\mathrm{prod}}
		=
		\frac12
		|\alpha\gamma+\beta\delta|^2.
	\end{equation}
	
	Applying the Cauchy--Schwarz inequality,
	\[
	|\alpha\gamma+\beta\delta|^2
	\le
	(|\alpha|^2+|\beta|^2)
	(|\gamma|^2+|\delta|^2)
	=
	1,
	\]
	one obtains
	\[
	F_{\mathrm{prod}}
	\le
	\frac12.
	\]
	
	Now consider a general separable mixed state
	\begin{equation}
		\rho_{AB}
		=
		\sum_k
		p_k
		\rho_A^{(k)}
		\otimes
		\rho_B^{(k)},
		\qquad
		p_k\ge0,
		\qquad
		\sum_k p_k=1.
	\end{equation}
	
	Since the Bell fidelity is linear in $\rho_{AB}$,
	\begin{equation}
		F
		=
		\sum_k
		p_k
		\bra{\Phi^+}
		\rho_A^{(k)}
		\otimes
		\rho_B^{(k)}
		\ket{\Phi^+}.
	\end{equation}
	
	Each term individually satisfies
	\[
	F_k\le\frac12,
	\]
	and therefore
	\begin{equation}
		F
		\le
		\sum_k
		p_k
		\frac12
		=
		\frac12.
	\end{equation}
	
	Consequently,
	\[
	F>\frac12
	\]
	constitutes a sufficient witness of bipartite entanglement.
	
	Combining this result with Eq.~(\ref{eq:p00_appendix}) immediately yields the operational certification condition
	\[
	P_{00}>\frac34,
	\]
	used throughout the main text.


\begin{thebibliography}{99}
		
		\bibitem{Kimble2008}
		H.~J.~Kimble,
		``The quantum internet,''
		Nature \textbf{453}, 1023--1030 (2008).
		
		\bibitem{Wehner2018}
		S.~Wehner \emph{et al.},
		``Quantum internet: A vision for the road ahead,''
		Science \textbf{362}, eaam9288 (2018).
		
		\bibitem{Azuma2023}
		K.~Azuma, A.~Pirker, D.~Markham, and W.~D\"ur,
		``Quantum repeaters: From quantum networks to the quantum internet,''
		Reviews of Modern Physics \textbf{95}, 045006 (2023).
		
		\bibitem{martin2024}
		V.~Mart\'in \emph{et al.},
		``MadQCI: a heterogeneous and scalable SDN-QKD network deployed in production facilities,''
		npj Quantum Inf. \textbf{10}, 80 (2024).
		
		\bibitem{Acin2007DIQKD}
		A.~Ac\'in \emph{et al.},
		``Device-Independent Security of Quantum Cryptography against Collective Attacks,''
		Phys. Rev. Lett. \textbf{98}, 230501 (2007).
		
		\bibitem{diamanti2016}
		E.~Diamanti, H.-K.~Lo, B.~Qi, and Z.~Yuan,
		``Practical challenges in quantum key distribution,''
		npj Quantum Inf. \textbf{2}, 16025 (2016).
		
		\bibitem{Watrous2018}
		J.~Watrous,
		\textit{The Theory of Quantum Information}
		(Cambridge University Press, Cambridge, 2018). ISBN: 9781316848142, 
		DOI: https://doi.org/10.1017/9781316848142
		
		
		\bibitem{Watts1998}
		D.~J.~Watts and S.~H.~Strogatz,
		``Collective dynamics of 'small-world' networks,''
		Nature \textbf{393}, 440--442 (1998).
		
		\bibitem{Acin2007}
		A.~Ac\'in, J.~I.~Cirac, and M.~Lewenstein,
		``Entanglement percolation in quantum networks,''
		Nature Physics \textbf{3}, 256--259 (2007).
		
		\bibitem{Cuquet2009}
		M.~Cuquet and J.~Calsamiglia,
		``Entanglement Percolation in Quantum Complex Networks,''
		Phys. Rev. Lett. \textbf{103}, 240503 (2009).
		
		\bibitem{Pant2019}
		M.~Pant \emph{et al.},
		``Routing entanglement in the quantum internet,''
		npj Quantum Information \textbf{5}, 25 (2019).
		
		\bibitem{Hastings2006}
		M.~B.~Hastings and T.~Koma,
		``Spectral gap and exponential decay of correlations,''
		Communications in Mathematical Physics \textbf{265}, 781--804 (2006).
		
		\bibitem{Verstraete2008}
		F.~Verstraete, V.~Murg, and J.~I.~Cirac,
		``Matrix product states, projected entangled pair states, and variational renormalization group methods for quantum spin systems,''
		Adv. Phys. \textbf{57}, 143--224 (2008).
		
		\bibitem{Horodecki1995}
		R.~Horodecki, P.~Horodecki, and M.~Horodecki,
		``Violating Bell inequality by mixed spin-$\frac12$ states: necessary and sufficient condition,''
		Physics Letters A \textbf{200}, 340--344 (1995).
		
		\bibitem{Pirandola2017}
		S.~Pirandola \emph{et al.},
		``Fundamental limits of repeaterless quantum communications,''
		Nature Communications \textbf{8}, 15043 (2017).
		
		\bibitem{Pironio2010}
		S.~Pironio \emph{et al.},
		``Random numbers certified by Bell's theorem,''
		Nature \textbf{464}, 1021--1024 (2010).
		
	\end{thebibliography}
\end{document}